\documentclass[11pt, dvipsnames]{article}
\usepackage{amsthm}
\usepackage{amssymb}
\usepackage{amsmath}
\usepackage{authblk}
\usepackage[a4paper, total={6.5in, 10in}]{geometry}
\usepackage[utf8]{inputenc}
\usepackage{graphicx}
\graphicspath{ {./images/} }
\usepackage{tikz}
\usetikzlibrary{positioning}
\usepackage{hyperref}
\hypersetup{
	colorlinks=true,
	citecolor=black,
	filecolor=black,
	linkcolor=black,
	urlcolor=blue
}

\usepackage{fancyhdr}
\usepackage{tcolorbox}
\usepackage{xcolor}
\usepackage[skip=0.5em, tocskip=0.5em, indent=2em, parfill = 0.5em]{parskip}
\usepackage{natbib}
\nocite{*}

\newtheorem{theorem}{Theorem}[section]
\newtheorem{lemma}{Lemma}[theorem]

\theoremstyle{definition}
\newtheorem{definition}{Definition}[section]

\newtheorem{example}[definition]{Example}

\theoremstyle{remark}

\renewcommand{\P}[1]{\mathbb{P}\left[#1\right]}

\newcommand{\E}[1]{\mathbb{E}\left[#1\right]}

\newenvironment{mech}[1]{\begin{tcolorbox}[colback=red!5!white,colframe=red!75!black,title={#1}]}{\end{tcolorbox}}

\title{Random Serial Dictatorship with Transfers}
\author{Sudharsan Sundar, Eric Gao, Trevor Chow, Matthew Ding \thanks{Stanford University.}}
\date{March 25, 2023}

\begin{document}

\maketitle

\begin{abstract}
    It is well known that Random Serial Dictatorship is strategy-proof and leads to a Pareto-Efficient outcome. We show that this result breaks down when individuals are allowed to make transfers, and adapt Random Serial Dictatorship to encompass trades between individuals. Strategic analysis of play under the new mechanisms we define is given, accompanied by simulations to quantify the gains from trade.
\end{abstract}


\tableofcontents
\newpage

\section{Motivation}

The question of how to assign unique, indivisible goods is a problem that is common and important in many economic contexts. One need look no further than a college campus to see this: the classic example is room assignments to students planning to live in on-campus dorms.

Most current approaches to this question, especially for dorm assignments, tend to follow a variant of Random Serial Dictatorship, a classic and useful mechanism in the market design literature. But, as any observant undergrad could tell you, aftermarkets, albeit informal and small in size, pop up in the wake of RSD. People willing to sell their rooms for some extra cash and those willing (and able) to buy for a high enough price conduct backroom (or, rather, `your-room, my-cash') deals in order to trade housing assignments.

This reveals an important consideration that has not been thoroughly appreciated when using RSD as an assignment mechanism: participants' welfare is heavily influenced by not only their room assignment but also the amount of money they walk away from the mechanism with. Thus, most analyses of assignment problems and RSD solutions miss the important gains from trade that could be realized if participants could easily and readily trade money for room assignments.

Our work hopes to investigate this area commonly missed in thinking about assignment mechanisms like RSD. The question we investigate is: what are the important properties of mechanisms that combine RSD with periods during which participants can engage in trades and transfers with one another, and how does this new type of mechanism compare to plain RSD?

\section{Related Works}

The allocation of indivisible goods such as houses is a well-studied problem in mechanism design, and one canonical approach is the randomised serial dictatorship (RSD) mechanism. The literature on RSD discusses a range of extensions and considerations that are not covered in the original mechanism. 

For example, one important aspect of problems like house allocation in real life is that they are often subject to endogenous information acquisition. This means that agents may have different information about the preferences of other agents, and this information may change over time. This is a key feature of many real-world problems, and it is important to consider how it affects the design of mechanisms.

\cite{bade_2015_serial} argues that while there are many optimal mechanisms for matching problems under perfect information, if we consider the process of endogenous information acquisition, then the unique ex ante Pareto-optimal, strategy-proof and non-bossy allocation mechanism is serial dictatorship. \cite{noda_2022_strategic} extends this analysis of RSD in an endogenous information acquisition setting, and shows that full disclosure about the choice sets of agents may not be entirely efficient, implying that there are positive externalities which may need to be considered in real life.

Others have considered how RSD compares with alternative randomised mechanisms. \cite{abdulkadiroglu_1998_random} demonstrate that the core from using a random endowment matching is equivalent to the results of RSD, and thus represent the same underlying lottery mechanism. \cite{bade_2020_random} extends this equivalence to the set of all Pareto-optimal, strategy-proof, and nonbossy matching mechanisms.

One further angle of research involves indivisible objects, and \cite{han_2016_on} generalises RSD to these objects by allowing for fractional distributions, and demonstrates how it is limited when there are more than 4 agents.

Finally, there is a lot of work on the ordinal efficiency of RSD, with \cite{manea_2009_asymptotic} showing that the probability that this is the case goes to 0 as the number of objects being matched increases, and in fact, \cite{manea_2007_random} finds that there are exchange contracts which dominate the RSD allocation.

However, one lacuna in the literature is incorporating monetary transfers after RSD. This is promising because \cite{klaus_2019_serial} highlight some normative distinctions when considering serial dictatorship with reservation prices and because \cite{hosseini_2015_on} demonstrate that outside of single-shot settings, RSD is manipulable.

Thus we now set out to explore the results of having monetary transfers with RSD. All omitted proofs can be found in Appendix B.

\section{Model and Baseline}

Our model consists of the following:

\begin{enumerate}
    \item A (finite) set of items $N$ indexed by $N = \{1,2,...,n\}$ with generic element $i \in N$;
    \item A (finite) set of agents $M$ indexed by $M = \{1,2,...,m\}$ with generic element $j \in M$;
    \item Each agent $j$ has quasilinear utility $u_j: N \cup \varnothing \times \mathbb{R} \to \mathbb{R}$ defined by $u_j(i,d) = v_j(i) + d$ where $v_j: N \cup \varnothing \to \mathbb{R}$ and $d$ can be intuitively thought of as ``dollars'';
    \item Each agent $j$ starts with $d_j$ dollars;
    \item An allocation $a$ is an injective (no two agents can match with the same item) function $M \to N \cup \varnothing$;
    \item A profile of transfers $t: M \to \mathbb{R}$ such that $\sum_{j \in M} t(j) = 0$ (money cannot appear or disappear). 
\end{enumerate}

Going forward, we can normalize $v_j(\varnothing) = 0$ for all $j$. Since each match and transfer profile induces some outcome, we can say that agents have preferences over matches and transfer profiles. Letting $A$ denote the set of all allocations and $T$ denote the set of all transfers, we can define $u_j': A \times T \to \mathbb{R}$ by
$$u_j'(a,t) = u_j(a(j), t(j)) = v_j(a(j)) + d_j + t(j).$$
Furthermore, we can (without much loss of generality) restrict to the case where the number of items is equal to the number of agents. In the motivating application of the Stanford housing market, if there are more agents (students) than items (houses), then there are larger institutional issues at hand than the question of allocation. However, if there are more items than agents, we can without loss ``remove'' any items that are not picked by independence of irrelavent alternatives. 

Our baseline algorithm is Serial Dictatorship, which does not include any transactions. 

\begin{mech}{(Random) Serial Dictatorship}
    Given a set of items $N$, a set of agents $M$, and each agent (privately) has preferences represented by $u_j$, do the following:
    \begin{enumerate}
        \setlength\itemsep{0.5em}
        \item Order the agents $M$ based on some arbitrary order;
        \item For each agent $j$ in $M$:
        \begin{enumerate}
            \item Ask agent $j$ to choose $i_j \in I$ maximizing their utility: $i_j \in \arg\max u_j(i,d_j)$;
            \item Set $a(j) = i_j$;
            \item Remove $i_j$ from $I$.
        \end{enumerate}
        \item Repeat until there are no agents or items left and set $a(j) = \varnothing$ for any remaining agents;
        \item Set $t(j) = 0$ for all agents.
    \end{enumerate}
    \vspace{0.5em}
    If the order in step one is by random, we call the mechanism Random Serial Dictatorship (RSD).
\end{mech}

Each run of RSD induces some allocation $a$ and the trivial transfer profile $t(j) = 0$. It is well known that in the absence of transfers, RSD leads to a Pareto optimal allocation: 

\begin{theorem}[Optimality of RSD]
    Suppose instead of quasilinear utilities, agents only care about the item they receive and $u_j(i,d) = v_j(i)$ for all $j$. Then, RSD always leads to a pareto efficient allocation $a$. That is, for all $a' \neq a$, there exists some $j$ such that
    $$v_j(a'(j)) < v_j(a(j)).$$
\end{theorem}

However, in the presence of transfers and utility over money, the above theorem is no longer true. 

\begin{example}
Suppose $N = \{x,y\}$ and $M = \{1,2\}$. Let $v_1(x) = 2, v_1(y) = 1, d_1 = 5$ and $v_2(x) = 10, v_2(y) = 1, d_2 = 5$. Then, the RSD outcome when $1$ picks before $2$ is: $a(1) = x, m(2) = y; t(1) = t(2) = 0$. Utilities are
$$u_1(a(1),t(1)) = v_1(x) + 5 + 0 = 7; u_2(a(2),t(2)) = v_2(y) + 5 + 0 = 6.$$
However, now consider the allocation $a'(1) = y, a'(2) = x, t(1) = 2, t(2) = -2$.\footnote{We could have taken $t(1) = n, t(2) = -n$ for any $n \in (1,9)$.} Now, 
$$u_1(a'(1),t(1)) = v_1(y) + 5 + 2 = 8; u_2(a'(2),t(2)) = v_2(x) + 5 -2 = 13.$$
Both individuals are better off, so the RSD outcome was not Pareto efficient.
\end{example}

For additional motivation to see why transfers are crucial, suppose we allowed trades to happen, but without transfers. This amounts to running TTC after RSD. We then get the following:

\begin{theorem}
    A Nash equilibrium of RSD followed by TTC is for each agent to choose their best available item in the RSD stage and for no trades to happen in the TTC stage. 
\end{theorem}

\begin{proof}
    We will prove this by induction on the spot in which an agent is listed in the RSD stage. The agent that picks first can never do better than matching with their best item, so they cannot have any profitable deviations away from picking their (unconstrained) top choice. Next, consider the agent in $n$th place and suppose all agents before them all choose their top available item. Let $a_n$ denote this agent.

    No matter what item $a_n$ chooses, they will never be able to trade with an agent that chose before them: $a_n$'s object was available for everyone that chose before them so by revealed preference, agents that chose before $a_n$ must prefer their item over any item that $a_n$ might have chosen. Thus, $a_n$'s utility is upper bounded by the utility of their favorite item, so there cannot be any profitable deviations from directly choosing it.
\end{proof}

\section{Ex-Post Transfers}

Allowing for Ex-Post Transfers adds an additional step to RSD. However, we need to first expand our definition of equilibrium to the case of endowments. Suppose agent $j$ starts with some endowed item $e(j)$. Then, we define:

\begin{definition}[Competitive Equilibrium with Endowments]
A competitive equilibrium with endowments is a price function $p: N \to \mathbb{R}$ and an allocation $a$ such that:
\begin{enumerate}
    \item each agent $j$ has 
    $$a(j) \in \arg\max_{i \in N} v_j(i) - p(i) + p(e(j));$$
    \item the set of items allocated is equal to the set of items that were in some agent's endowment.\footnote{This second condition is standard in general equilibrium theory and (with some minor assumptions) is equivalent to the condition that un-allocated items have a price of zero as an item will be allocated if and only if someone chose it in the first stage of RSD.}
\end{enumerate}
\end{definition}

Then, we have the following:

\begin{mech}{(Random) Serial Dictatorship with Ex-Post Transfers}
    Given a set of items $N$, a set of agents $M$, and each agent (privately) has preferences represented by $u_j$, do the following:
    \begin{enumerate}
    \setlength\itemsep{0.5em}
        \item Order the agents $M$ based in some arbitrary order and set $I_1 = I$;
        \item For each agent $j$ in $M$:
        \begin{enumerate}
            \item Ask agent $j$ to choose $i_j \in I_j$;
            \item Set $e(j) = i_j$;
            \item Set $I_{j+1} = I_j \setminus \{i_j\}$.
        \end{enumerate}
        \item Repeat until there are no agents or items left and set $a(j) = \varnothing$ for any remaining agents;
        \item Set $t(j) = 0$ for all agents;
        \item Allow agents to trade until a competitive equilibrium with endowments is reached. If $p:N \to \mathbb{R}$ maps items to their price in equilibrium and $a$ is the allocation of goods in equilibrium, set $t(j) = p(e(j)) - p(a(j))$ to be the value of agent $j$'s endowment minus the cost of agent $j$'s house in equilibrium.
    \end{enumerate}
\end{mech}

One note for completeness: to make sure that $(a,t)$ is a valid outcome, we need to ensure that $\sum_{j \in M} t(j) = 0$. This is true, as 
$$\sum_{j \in M} t(j) = \sum_{j \in M} p(e(j)) - \sum_{j \in M} p(a(j)) = \sum_{i \in range(e)} p(i) - \sum_{i \in range(e)} p(i) = 0$$
where $range(e)$ is the set of items that are in some agent's endowment.\footnote{Mathematically, this is just the range of $e$ when viewed as a function from agents to items.} The final equality is true as both the endowment and allocation must distribute the same goods by the second condition of Competitive Equilibrium with Endowments. It turns out that Competitive Equilibrium is a powerful tool in this environment:

\begin{theorem}[Equilibrium is Pareto-Efficient]
Every outcome of RSD with Ex-Post Transfers must be Pareto-Efficient.
\end{theorem}

\begin{proof}
    We will show an equivalence between the definition of Competitive Equilibrium with Endowments and the usual definition of Competitive Equilibrium. First, as $p(e(j))$ is independent of $i$, we have that
    $$\arg\max_{i \in N} v_j(i) - p(i) + p(e(j)) = \arg\max_{i \in N} v_j(i) - p(i)$$
    so 
    $$a(j) \in \arg\max_{i \in N} v_j(i) - p(i) + p(e(j)) \text{ if and only if } a(j) \in \arg\max_{i \in N} v_j(i) - p(i).$$
    Next, the set of goods that are in someone's endowment is exactly equal to the set of goods that are demanded, so 
    $$p(i) = 0 \text{ if and only if } i \text{ is not assigned}.$$
    With these two equivalences, the proof of why competitive equilibrium is optimal without endowments follows. 
\end{proof}

A direct corollary is that the RSD with Ex-Post Transfers allocation of items to agents is exactly the same as the VCG allocation of items to agents.

\begin{proof}
    Both RSD with Ex-Post Transfers and VCG chooses the allocation that maximizes total societal welfare.
\end{proof}

One possible limitation about the above result is that getting to equilibrium in step 5 is a black box. Suppose we replaced step 5 in RSD with Ex-Post transfers with a restriction to pairwise transfers:

\begin{mech}{(Random) Serial Dictatorship with Ex-Post Pairwise Transfers}
    Given a set of items $N$, a set of agents $M$, and each agent (privately) has preferences represented by $u_j$, do the following:
    \vspace{0.5em}
    \begin{enumerate}
    \setlength\itemsep{0.5em}
        \item Order the agents $M$ based in some arbitrary order and set $I_1 = I$;
        \item For each agent $j$ in $M$:
        \begin{enumerate}
            \item Ask agent $j$ to choose $i_j \in I_j$;
            \item Set $e(j) = i_j$;
            \item Set $I_{j+1} = I_j \setminus \{i_j\}$.
        \end{enumerate}
        \item Repeat until there are no agents or items left and set $a(j) = \varnothing$ for any remaining agents;
        \item Set $t(j) = 0$ for all agents;
        \item Allow agents to trade. If $k$ is willing to pay $p$ to trade with $j$ then set:
         \begin{itemize}
            \item $a'(k) = a(j)$ and $a'(j) = a(k)$ where $a'$ is the new allocation after the trade;
            \item $t'(k) = t(i) - p$ and $t'(j) = t(k) + p$ where $t'$ is the new transfer after the trade.
        \end{itemize}
    \end{enumerate}
\end{mech}

This interpretation is much closer to what happens in the Stanford housing aftermarket: individuals make pairwise agreements to trade houses instead of a centralized market forming. Do our optimally results still hold? Our next example provides a negative answer that uses the following Lemma:

\begin{lemma}[Characterization of Trades]
\label{trade1}
    A trade between agent $j$ with item $a(j)$ and agent $k$ with item $a(k)$ is possible if and only if 
    $$v_j(a(k)) + v_k(a(j)) > v_j(a(j)) + v_k(a(k)).$$
\end{lemma}

\begin{example}[Inefficiency of Pairwise Transfers]
    Suppose there are three agents $A,B,C$ and items $1,2,3$. Payoffs are described in the following matrix: 
    
    \begin{table}[h!]\centering
    \begin{tabular}{c|c|c|c|}
      & 1 & 2 & 3 \\ \hline
      $A$ & 5 & 0 & 10 \\ \hline
      $B$ & 0 & 4 & 0 \\ \hline
      $C$ & -10 & 0 & 5 \\ \hline
    \end{tabular}
    \end{table}
    
    where, for example, the $4$ in row $B$, column $2$ corresponds to agent $B$ gaining a utility of $4$ from item $2$. Suppose RSD with pairwise transfers is run with $A$ picking first, then $C$, and finally $B$. First, $A$ picks item $3$; then $C$ picks item $2$, and finally $B$ is left with item $2$. Thus, the allocation before trade is $a(A) = 3, a(B) = 1, a(C) = 2$. Now, note that no pairwise trades are possible: 
    \begin{enumerate}
        \item $A$ and $B$ cannot trade as $v_A(a(B)) + v_B(a(A)) = 5+0 < 0+10 = v_A(a(A)) + v_B(a(B))$;
        \item $A$ and $C$ cannot trade as $v_A(a(C)) + v_C(a(A)) = 5+0 < 10+0 = v_A(a(A)) + v_C(a(C))$;
        \item $B$ and $C$ cannot trade as $v_B(a(C)) + v_B(a(A)) = 0+0 < 4-10 = v_B(a(B)) + v_C(a(C))$.
    \end{enumerate}
    Thus, no trades are possible from this allocation. However, the efficient allocation is $a(A) = 1, a(B) = 2, a(C) = 5$. This is reachable through a ``joint trade'' of all three agents (for instance, $B$ pays $3$ to $A$ and $C$ pays $3$ to $A$ to make this happen).
\end{example}

In general, this counterexample holds as $v_j(a(j)) + v_k(a(k)) \geq v_j(a(k)) + v_k(a(j))$ for all $j,k \in M$ does not imply Pareto efficiency. However, we do have a weakened version of Theorem 3.1:

\begin{theorem}[Transfers are Pareto-Improving]
The outcome after step 5 of RSD with Ex-Post Transfers will always be a Pareto improvement over the outcome after step 4. 
\end{theorem}

\begin{proof}
Every pairwise trade is a Pareto improvement since it must be mutually beneficial to both parties involved. Thus, it suffice to show that the number of possible trades is finite as then inducting on the number of trades gives the desired result. Towards a contradiction, suppose there are an infinite number of trades made. In each trade, at least one person must be strictly better off. Furthermore, since the set of values over all agents and items is a finite discrete set, there is some lower bound on the gain in utility from each trade. As such, if there are an infinite number of trades, societal utility would become arbitrarily large, which cannot be the case. As such, there can only be a finite number of trades, and the proof is complete.
\end{proof}

That being said, in the special case of all agents having the same preferences over items, we see that RSD with Ex-Post pairwise transfers is equivalent to plain RSD.

\begin{theorem}[Special Case of Identical Preferences]
    If all agents have the same preferences over items, RSD with ex-post pairwise transfers is equivalent to using plain RSD.
\end{theorem}

\begin{proof}
Assume we are matching N items to M agents using the mechanism RSD with ex-post pairwise transfers, and assume all agents have the same preferences over items. Assume that if at least one agent prefers trading to not trading and the other agent is at least indifferent between trading and not trading then the two agents trade.

Since all agents have the same preferences, no agents will want to trade from any initial endowment.

To see this, consider for contradiction that there are some agents A and B endowed with items 1 and 2, respectively, that are both willing and able to trade.

Without loss of generality, we know that $v_A(1) = v_B(1) = x, v_A(2) = v_B(2) = y, x \geq y$. Thus, agent A must receive a payment $p > x-y$, and agent B will only pay $p' \leq x-y$ (or agent A must receive a payment $p \geq x-y$, and agent B will only pay $p' < x-y$). This is a contradiction: $p \not= 'p$, so there is no possible trade between agents A and B.

Thus, since agents know no trade will occur during the transfer period, the strategic situation reduces to become equivalent to plain RSD: since the item agents end up with is the item they choose in the RSD portion of the mechanism, they behave identically to using the plain RSD mechanism.

\end{proof}

An interesting conclusion from this characteristic of RSD with ex-post transfers is that sufficiently aligned preferences among agents should lead the mechanism to function very similarly to plain RSD.

Another interesting and characteristic property of RSD with ex-post transfers is the fact it is not strategy-proof. To see this, consider the below example.\\

\begin{example}[Not Strategy-Proof]

    Let the mechanism being used be RSD with ex-post transfers, allocating the set of N items among M agents. Let agent A place a value of 10 utility on all rooms except room 1, which he values at 20 utility. Let agent B place a value of 10 utility on all rooms except room 2, which she values at 200 utility. Let all other agents place a value of 10 utility on all rooms.

    Say A picks first and has a budget of 0 dollars, B picks second and has a budget of 500 dollars, A knows B's budget and preferences, and B knows A's budget and preferences. All other agents choose rooms randomly (since they're indifferent between receiving any room). Consider A's strategy: 
    \begin{itemize}
        \item A chooses room i, $i \not= 1, 2$. A's payoff is $10$: all other agents value room i at 10 and receive a room worth at least 10 to them, so there will be no buyers who will pay more than 0 to get the room and A will only sell for more than 0.
        \item A chooses room 1. A's payoff is $20$: all other agents value room 1 at 10 and receive a room worth at least 10 to them, so there will be no buyers who will pay more than 0 to get the room and A will only sell for more than 10.
        \item A chooses room 2. B chooses room 1 to gain greater leverage for trading with A. A's payoff will be $20 + t$, $0 \leq t < 190$: B would only be willing to buy from A for less than 190 (receiving a room worth 200, minus trading away a room worth 10 and making some payment t) and A would only be willing to sell for more than or equal 0. B doesn't offer to sell room 1 to A since A has 0 budget. Thus, A's payoff will be $20 + t \geq 20$.
    \end{itemize}

    Thus, we see that A can gain (immensely) from acting strategically rather than honestly: if A just charges a 'half-way' price in the third scenario, then $t = 95$. Generalizing this example, we see that RSD with ex post transfers is very far from strategy-proof--in fact, there can be very strong incentives to act strategically, specifically to buy `popular' rooms and sell high on the transfer market.

\end{example}

\section{Strategic Analysis of Ex-Post Transfers and the Interim Case}

We've seen that adding transfers makes RSD no longer strategy-proof. A natural question is what actually happens. To that end, this section will try to characterize equilibrium behavior. However, we need the following assumption to make analysis tractable. Suppose people's preferences over all rooms are revealed as soon as they select their room, and no strategizing is possible over people who have not selected yet.\footnote{To provide some justification, an alternative assumption that implies that no strategizing is possible over people who have not yet selected is to assume that individuals are extremely risk-adverse. As such, they would prefer any decent room with certainty than to gamble over potential outcomes. On the other hand, if there are two people who have entered into an agreement with certainty, this agreement is equivalent to swapping their places, and the same analysis holds.} 

We will now analyze the subgame-perfect Nash equilibrium of RSD with transfers. After RSD is run, we assume that there is no strategizing in the transfers phase.\footnote{Analyzing strategy in this phase goes beyond the current scope of the paper. However, some related literature could include uniqueness of general equilibrium (\cite{sandberg_1979_uniqueness}) or bargaining (\cite{rubinstein_1982_perfect}).} As such, we focus on the distortionary effect transfers has on how choices within RSD are made. 

Proceeding via backwards induction, the agent who is assigned to choose last has only one choice of item, namely the one item left over. Next, consider an agent that picks from a (nontrivial) menu of houses. There are several possible ways for them to pick:
\begin{enumerate}
    \item Pick an item to maximize utility from that item itself;
    \item Pick an item that an agent who has already chosen likes to try and trade with them;
    \item Pick an item and hope to trade with someone that picks after you.
\end{enumerate}
We can immediately rule out the third choice: by assumption, agents will not choose an item expecting an agent that chooses later to buy it. Similarly, agents will not pay for an item allocated to an agent that chooses later as they could simply choose the second item themselves and pay nothing.

The first option is straightforward: the agent compares all available options and chooses their favorite item. However, picking an item while anticipating selling it could be a better option: if this is the case, we get the following:

\begin{lemma}
\label{trade}
    Suppose at agent $j$'s turn to pick, they choose item $i$ to trade with agent $\hat{j}$ who currently has item $\hat{i}$. Then, if $I_j$ is the set of items available when it is $j$'s turn to pick, it must be that
    $$i \in \arg\max_{i' \in I_j} v_{\hat{j}}(i').$$
\end{lemma}

\begin{proof}
    For a trade to happen, $j$ needs to pay $\hat{j}$ some $t > 0$ so to simplify notation, let $t(j) = -t$ and $t(\hat{j}) = t$. Then, for the trade to be agreeable for $\hat{j}$, we need that
    $$v_{\hat{j}}(i) + t \geq v_{\hat{j}}(\hat{i}) \implies  t \geq  v_{\hat{j}}(\hat{i}) - v_{\hat{j}}(i)$$
    so agent $j$'s maximization problem is
    $$\max_{t, i' \in I_j} v_j(\hat{i}) - t \text{ s.t. } t \geq  v_{\hat{j}}(\hat{i}) - v_{\hat{j}}(i).$$
    As such, we can get an upper bound on agent $j$'s utility to be
    $$\max_{t, i' \in I_j} v_j(\hat{i}) - t \text{ s.t. } t \geq  v_{\hat{j}}(\hat{i}) - v_{\hat{j}}(i) \leq \max_{i' \in I_j} v_j(\hat{i}) - [v_{\hat{j}}(\hat{i}) - v_{\hat{j}}(i)] = v_j(\hat{i}) - v_{\hat{j}}(\hat{i}) + \max_{i' \in I_j} v_{\hat{j}}(i').$$
    As such, choosing $i'$ that maximizes $v_{\hat{j}}(i')$ is agent $j$'s best option. Furthermore, for any fixed allocation of the gains from trade between $j$ and $j'$, if $\arg\max_{i' \in I_j} v_{\hat{j}}(i')$ is unique then choosing the maximizer is a strictly dominant strategy.
\end{proof}

While this result is intuitive, how applicable is it for practical analysis of RSD with ex-post transfers? A major concern is whether or not the agent $j$ is trying to trade with will receive a better offer from some other agent $j^*$. As such, the prior Lemma is not enough to fully characterize subgame-perfect behavior. In general, if $j$ is faced with choice set $I_j$, each $i \in I_j$ will induce some expected payoff, and $j$ will choose the item that maximizes their expected payoff. Another way to resolve this uncertainty is to incorporate trades into the RSD process itself. We define the following mechanism:

\begin{mech}{(Random) Serial Dictatorship with Interim Transfers}
    Given a set of items $N$, a set of agents $M$, and each agent (privately) has preferences represented by $u_j$, do the following:
    \vspace{0.5em}
    \begin{enumerate}
    \setlength\itemsep{0.5em}
       \item Order the agents $M$ based in some arbitrary order, set $t(j) = 0$ for all agents, and set $I_1 = I$;
        \item For each agent $j$ in $M$:
        \begin{enumerate}
            \item Ask agent $j$ to choose $i_j \in I_j$;
            \item Set $a(j) = i_j$;
            \item Set $I_{j+1} = I_j \setminus \{i_j\}$;
            \item Allow agent $j$ to offer a trade with any agent $j' < j$ at a price of $t$;
            \begin{itemize}
                \item If the trade is accepted, swap $a(j)$ and $a(j')$ and set $t(j) = t(j)-t, t(j') = t(j')+t$;
                \item If the trade is not accepted or agent $j$ does not offer a trade, do nothing.
            \end{itemize}
        \end{enumerate}
        \item Repeat until there are no agents or items left and set $a(j) = \varnothing$ for any remaining agents;
    \end{enumerate}
\end{mech}

By construction, only one offer is made at a time, so combined with the assumption that agents do not strategize over future agents, there will no longer be cases where one offer is rejected in favor of a better offer. A structural note about this mechanism: observe each agent can only make a single trade. This is without loss, as the set of (feasible) trades attains a maximum via the logic of Lemma \ref{trade} so each agent will only need to make a single offer. 

Unfortunately, interim transfers does not resolve the possibility of failing to trade to the societally optimal allocation. Consider the following example:\footnote{Courtesy of feedback from the teaching team.}

\begin{example}[Inefficiency of Interim Transfers]\label{ineff}
    Suppose there are three agents $A,B,C$ and items $1,2,3$. Payoffs are described in the following matrix: 
    
    \begin{table}[h!]\centering
    \begin{tabular}{c|c|c|c|}
      & 1 & 2 & 3 \\ \hline
      $A$ & 10 & 9 & 0 \\ \hline
      $B$ & 0 & 10 & 9 \\ \hline
      $C$ & 4 & 0 & 1 \\ \hline
    \end{tabular}
    \end{table}
    
    Then, the allocation that maximizes societal welfare is $(A,2);(B,3);(C,1)$. However, running serial dictatorship with interim transfers leads to $A$ picking $1$ and no trades; $B$ picking $2$ with no trades, and $C$ picking $3$ with no trades. 
\end{example}

A sufficient condition to reach the allocation that maximizes societal welfare is for a path to achieving that to exist: if it is feasible to achieve the optimal allocation, then every individual acting in their own self-interest gets to that allocation.

\begin{theorem}
\label{interim}
    Suppose that when any agent goes to choose their item, it is feasible for them to induce the optimal allocation up to that point. Then, RSD with Interim Transfers always yields the allocation that maximizes total societal utility.
\end{theorem}

Fortunately for practical application (but unfortunately for theoretical sharpness), the feasibility assumption is not a necessary condition: consider the same example from before, but with an added agent and item:

\begin{example}[Interim Transfers May Work]
    Suppose there are four agents $A,B,C,D$ and items $1,2,3,4$. Payoffs are described in the following matrix: 
    
    \begin{table}[h!]\centering
    \begin{tabular}{c|c|c|c|c|}
      & 1 & 2 & 3 & 4\\ \hline
      $A$ & 10 & 9 & 0 & 0\\ \hline
      $B$ & 0 & 10 & 9 & 0\\ \hline
      $C$ & 4 & 0 & 1 & 0\\ \hline
      $D$ & 0 & 0 & 100 & 0\\ \hline
    \end{tabular}
    \end{table}
    
    Then, the allocation that maximizes societal welfare is now $(A,1);(B,2);(C,4), (D,3)$. Running Serial Dictators with interim transfers leads to the same outcome as Example \ref{ineff} for agents $A,B,C$. However, $D$ can choose item $4$ and pay agent $C$ a sufficiently high amount to induce a trade. Then, the resulting allocation is $(A,1);(B,2);(C,4);(D,3)$ which maximizes societal welfare.
\end{example}

Intuitively, future agents potentially can correct for inefficiencies by making correctional trades. A question for further research would be to find some necessary and sufficient condition on agent utilities that takes this into account and sharply characterizes when interim transfers induce the allocation that maximizes societal welfare.

\section{Two Agents and Items}

Clearly, the informational assumptions made in the last section (perfect information about agents that have already chosen and no strategizing about agents that choose afterwards) are not realistic. This section presents analysis under some different simplifying assumptions. 

Without loss of generality, let agent $j$ be the agent that chooses $j$th. Let $F_j^i$ be a cumulative distribution function for agent $j$'s value for item $i$, so 
$$F_j^i(x) = \P{v_j(i) \leq x}.$$
Assume that for all agents $j$ and items $i$, we have that $F_j^i$ is continuous, has compact support, and there exists $\underline{V}, \overline{V}$ such that $F_j^i(\underline{V}) = 0, F_j^i(\overline{V}) = 1$. In the information structure for this section, all $F_j^i$'s are commonly known, but each agent only knows their own valuation. As such, the whole timing of the game is as follows:
\begin{enumerate}
    \item Nature reveals to each agent their private valuations of each item;
    \item Serial Dictatorship with Interim Transfers is run;
    \item Payoffs are realized.
\end{enumerate}
Suppose $N = \{A,B\}$ and $M = \{1,2\}$ so there are two items and two agents. What is the equilibrium of this game?

We will proceed using backwards induction. Consider the subgame after $1$ has chosen. Player two only has one item available, so they must choose that. Suppose $2$ gets item $B$ (the other case follows symmetrically). Then, $2$ has the choice to offer a trade or not. If $v_2(B) \geq v_2(A)$ then $2$ already has their preferred item and has no incentive to trade. Otherwise, suppose $v_2(B) < v_2(A)$. If $2$ offers $1$ a transfer of $t$ to trade rooms, $1$ will accept the trade if
$$v_1(B) + t \geq v_1(A).$$
By definition, $v_1(B) \sim F_1^B$ and $v_1(A) \sim F_1^A$ so
\begin{equation*}
    \begin{split}
        \P{v_1(B) + t \geq v_1(A)} &= \P{v_1(B) \geq v_1(A) - t} \\
        &= 1 - \P{v_1(B) < v_1(A) - t} \\
        &= 1 - \E{F_1^B(v_1(A) - t)} \\
        &= 1 - \int_{\underline{V}}^{\overline{V}} F_1^B(s - t)dF_1^A(s)
    \end{split}
\end{equation*}
where the expectation in line three is taken with respect to $v_1(A)$ distributed according to $F_1^A$. As such, $2$'s expected payoff offering a transfer of $t$ is\footnote{As a nice sanity check, note that if $v_2(A) = v_2(B)$ then the (unique) optimal trade is $t = 0$, which makes sense.}
$$\pi_2(t) = (v_2(A)-t) \cdot \left[1 - \int_{\underline{V}}^{\overline{V}} F_1^B(s - t)dF_1^A(s)\right] + v_2(B) \cdot \int_{\underline{V}}^{\overline{V}} F_1^B(s - t)dF_1^A(s).$$
By continuity of $F_1^B$, we get that $2$'s expected profit is continuous in $t$. Furthermore, we can without loss of generality restrict $t$ to be in 
$$\left[-|\underline{V}-\overline{V}|,|\underline{V}-\overline{V}|\right]$$
since it is impossible for a trade to be mutually beneficial outside of this range. This is a compact set, so $\pi_2$ achieves a maximum over the relevant interval and there exists some $t^*$ satisfying
$$\pi_2(t^*) = \max_{t' \in \left[-|\underline{V}-\overline{V}|,|\underline{V}-\overline{V}|\right]} \pi_2(t').$$
Before moving onto the next step of the backwards induction, we present an intuitive comparative statics result:

\begin{theorem}
\label{compstat}
    Player $2$'s optimal offer $t^*$ is weakly increasing in the potential gain from trade, $v_2(A) - v_2(B)$.
\end{theorem}

With additional differentiability assumptions on the probability distributions involved, it may be possible to characterize the optimal trade offer $t^*$ using a first-order condition approach. However, such a computation will heavily depend on the probability distributions themselves and will be intractable beyond providing comparative statics, which we have already done.

Recall that the value of $t^*$ was defined assuming that $2$ is left with item $B$ once $1$ chooses. Furthermore, it is in terms of $v_2(A)$ and $v_2(B)$. To generalize, let $t^*(i,v_2(A),v_2(B))$ be $2$'s optimal transfer offer if they receive item $i$, value item $A$ at $v_2(A)$, and value item $B$ at $v_2(B)$ conditional on them wanting to make an offer at all.

If $1$ picks item $A$ and $2$ picks item $A$, 
\begin{enumerate}
    \item With probability 
    \begin{equation*}
        \begin{split}
            \P{v_2(B) > v_2(A)} &= \int_{\underline{V}}^{\overline{V}} F_2^A (s) dF_2^B(s)
        \end{split}
    \end{equation*}
    agent $2$ prefers $B$ to $A$ anyways and offers no trade;
    \item With complementary probability, $2$ offers $1$ a trade $t$ distributed according to $T_2^B$ where 
    $$T_2^B(s) = \P{t^*(B, v_2(A), v_2(B) < s}$$ 
    and the probability is taken with respect to $v_2(A) \sim F_2^A$ and $v_2(B) \sim F_2^B$.
\end{enumerate}
Thus, $1$'s expected utility from choosing $A$ is 
\begin{equation*}
    \begin{split}
        &\int_{\underline{V}}^{\overline{V}} F_2^A (s) dF_2^B(s) \cdot v_1(A) +\\
        &\left[1-\int_{\underline{V}}^{\overline{V}} F_2^A (s) dF_2^B(s)\right] \cdot \left[v_1(A) \cdot T_2^B(v_1(A) - v_1(B)) \right] + \\
        &\left[1-\int_{\underline{V}}^{\overline{V}} F_2^A (s) dF_2^B(s)\right] \cdot\left[ \E{v_1(B)+t|t>v_1(A) - v_1(B)}\cdot [1-T_2^B(v_1(A) - v_1(B))] \right]
    \end{split}
\end{equation*}
where the first line is the case where no offer by $2$ is made, the second line is the case where an offer is made by $2$ but is not satisfactory, and the third line is the case where an offer is made by $2$ and is accepted. By symmetry, swapping all instances of $A$ and $B$ in the above equation gives $1$'s expected utility from choosing $B$. Then, agent $1$ chooses the larger of the two values in any subgame perfect Nash equilibrium.

\section{Simulation}

To validate our theoretical results, we design and analyze a simulation of RSD with ex-post transfers in various settings. The design of the simulation is modeled after the Stanford undergraduate housing allocation problem, and we analyze the welfare generated by the novel mechanism compared to plain RSD, the distribution of welfare based on initial endowments, and the effects of transaction costs, among other things. 

The details of our simulation are as follows:

\begin{itemize}
\item We use 10,000 agents and 10,000 rooms.
\item We assign each room a normal distribution with mean and variance uniformly randomly chosen from $\mu\in[100, 10000]$ and $\sigma^2=[500, 1000]$. Agents' private valuations of a given room are randomly drawn from that room's distribution; for each agent $j$ and room $i$, define $v_j^{priv}(i)$ to be agent $j$'s private valuation of room $i$.
\item We fix a random ordering of the agents that is used as the pick order for RSD and the proposing order for ex-post pairwise transfers.
\item We allocate rooms during the RSD portion of the mechanism by providing agents 'incomplete but useful' information on others' preferences and allowing them to act strategically based on this information.
\begin{itemize}
    \item We define the public valuation for each room $i$ as $p^{pub}(i)$ to be the private valuation of the agent assigned to it after running plain RSD using the aforementioned fixed random ordering.
    \item We define each agent $j$'s augmented valuation of room $i$ to be \[v^{aug}_{j}(i)=\max\left(v^{priv}_{j}(i), \frac{v^{priv}_{j}(i)+v^{pub}(i)}{2}\right)\]
    \item We run plain RSD, but with agents acting based on their augmented valuations of rooms (as opposed to their private preferences) to generate the initial allocation of housing.
\end{itemize}
\item We then allow for pairwise transfers between agents. 
\begin{itemize}
    \item Each agent $j$ considers all rooms that they value more than their room assignment after RSD, $i_j$.
    \item For each of these more desirable rooms $i_k$ with owner $k$, the buy price is set to be the difference in agent $k$'s (the current owner's) valuations of the two rooms $p_{j,k}=v_k(i_k)-v_k(i_j)$.
    \item Agent $j$ computes their utility as $u_{j,k}=v_j(i_k)-v_j(i_j)-p_{j,k}$ and pays for the room that maximizes this utility, if at all any are positive.
    \item We assume that once an agent buys a room, they no longer participate in other pairwise transfers, i.e. they will not sell their room after they've bought it.
\end{itemize}
\end{itemize}

We use 10,000 agents and rooms since it's on the same order of magnitude of the real housing allocation problem at Stanford while still allowing for efficient simulation.

Our method of generating private valuations uses sampling from a common normal rather than a uniform since people's preferences for certain rooms tend to be highly correlated: pretty much everyone highly values a luxurious two room double in Toyon, and pretty much everyone dislikes a dingy one room double in Crothers Hall. That being said, people still have some meaningful differences in terms of how much they value a given room, and the relatively high variance of the distribution captures this element of reality.

A defining characteristic of this system is its lack of strategyproofness, and, thus, we provide agents with `incomplete but useful' information on others' preferences for them to act strategically during the RSD portion of the mechanism. We achieve this `incomplete but useful' information by effectively giving agents 'one more sample' from each room's distribution by running Plain RSD, but which is biased to account for the choosing order. We implement this feature of the simulation by performing plain RSD using "augmented preferences," as defined above. We take the maximum of the private and the average of private and public valuations since each agent is guaranteed their private valuation (there is no need for an agent to sell) and, depending on a room's resale value, might derive additional benefit from choosing to sell the room. We average private and public valuations as a simple method of incorporating uncertainty about being able to trade the room for its public value on the transfer market.

We use a pairwise trading system for transfers due to its realistic nature: practically speaking, coordinating a group 'cyclic' trade is difficult to execute, so doesn't factor importantly into our analysis of the behavior of the system (as discussed earlier). The added constraint of not allowing agents to participate in transfers once they buy a room also matches the real-world case: intuitively, the vast majority of participants in the real world participate in about one trade at most, not desiring to undertake the endeavor of 'high frequency trading on the dorm housing transfer market'.

Our simulation code can be found at \href{https://github.com/mattyding/RSD-w-transfers/tree/main}{this link}. 

\section{Welfare Analysis}
\label{sec-welfare-analysis}

We switch now to one of the major tests of the usefulness of a mechanism: its effects on the total welfare of the market. To that end, we analyze the total welfare generated by RSD with ex-post transfers through theoretical analysis and our simulation of the mechanism.

Given ideal theoretical conditions, we see that RSD with ex-post transfers performs very well when maximizing welfare.
\begin{theorem}
    \label{thm-expost-always-greater-than-rsd}
    Assuming the transfer market permits group trades, there are no transaction costs, and agents have quasilinear utility over item assignment and 'dollars,' RSD with Ex-Post transfers always generates greater total welfare than plain RSD.
\end{theorem}

\begin{proof}
    As proven earlier, we know RSD with ex-post transfers maximizes total social welfare. Thus, the total welfare generated by plain RSD must be less than or equal to the total welfare attained by RSD with ex-post transfers, holding constant the number of agents, the items allocated, agent preferences, and choosing order.
\end{proof}

Hence, we see that we get a (idealized) theoretical guarantee of the performance of RSD with ex-post transfers being better than plain RSD.

In order to analyze the mechanism under more realistic conditions, we also ran approximately 20 independent simulations to directly compute the increase in welfare from RSD with ex-post transfers relative to plain RSD, with each agent's final welfare defined as the sum of their remaining budget and their valuation of their room assignment by the end of the mechanism. All agents were given an initial endowment of $\$10,000$. The results are shown in figure 1.

\begin{figure}[h!]
    \centering
    \includegraphics[width=0.65\textwidth]{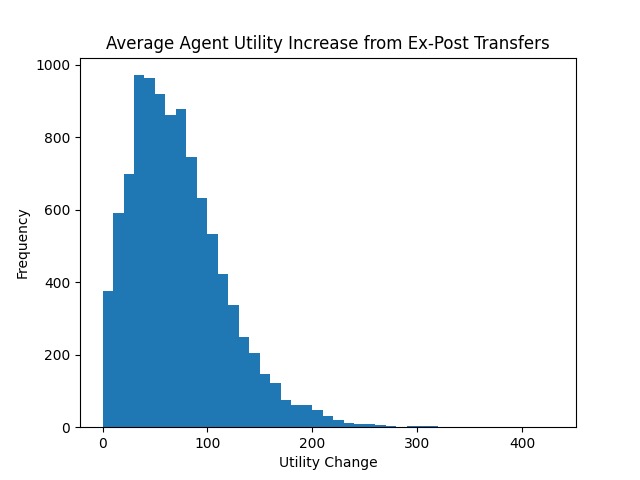}
    \caption{Welfare Increase of RSD with Ex-Post Transfers over Plain RSD}
    \label{fig:welfare}
\end{figure}

Corroborating Theorem \ref{thm-expost-always-greater-than-rsd} in more realistic conditions, no agent had lower welfare in the ex-post transfer case than in plain RSD. We also find that individual agent welfare increases follow a roughly left-skewed normal distribution, which may arise from agent valuations themselves being drawn from a normal distribution.

Though these results are reassuring, the assumption of equal budgets for all participants is likely far from reality and may lead to an inaccurate view of the mechanism's efficiency. Thus, in order to account for the influence of socio-economic inequality on the model's functioning, a likely element of its real world application, we once again simulated the total and individual welfares generated by RSD with ex-post transfers relative to plain RSD, but this time with agent budgets following a power law distribution, which accords with real world observations of income distribution, such as the well known observation of the Pareto distribution of incomes within a nation. The specific wealth distribution used in our experiment is we created 1000 income groups, where the initial endowment of $\texttt{group}_i = 1.01^i$, and each income group had 10 agents assigned to it. The results of this experiment are shown in Figure 2.

\begin{figure}[h!]
    \centering
    \includegraphics[width=0.49\textwidth]{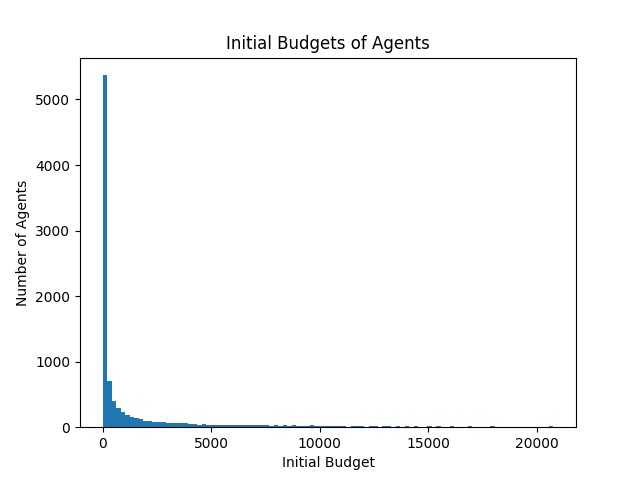}
    \includegraphics[width=0.49\textwidth]{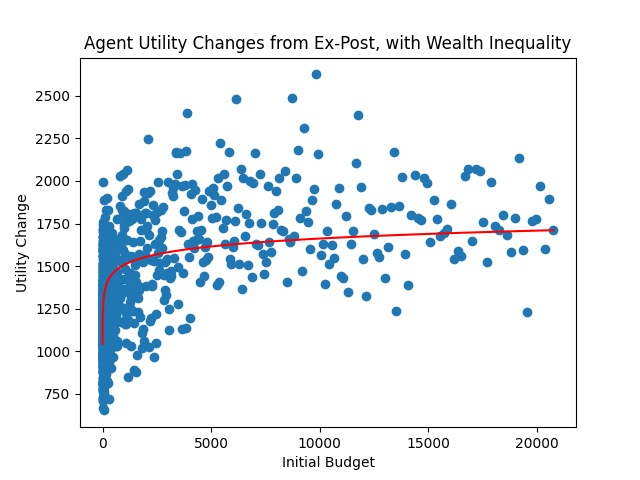}
    \caption{Welfare Analysis with Wealth Inequality}
    \label{fig:wealth-inequality}
\end{figure}

In this case we found the individual change in utility scaled logarithmically with the agent's initial budget. Thus, the benefit that higher resource participants receive quickly diminishes in importance as the agent becomes wealthier, rendering them similarly well off to those with `middle class' endowments.

\section{Transaction Costs}

An important consideration when analyzing mechanisms with transfers is the cost of learning about and executing trades in a real world setting with imperfect and perhaps costly information. Hence, we test our model's practically by including transaction costs to our RSD with Ex-Post Transfers model:

\begin{enumerate}
    \item Let $\tau \in \mathbf{R}$, the amount of the transaction cost;
    \item All transactions face a `tax' of $\tau$ either proportional to the price of a transaction or set to a fixed sum. This tax is subtracted from the total utility an agent receives from their outcome in the mechanism. Only one side of the trade incurs this cost (i.e. the seller's side).
    \item All agents know exactly what the level and nature of the transaction cost is and incorporate this information in forming their strategies.
\end{enumerate}

With this definition, we see that plain RSD is equivalent to RSD with Transfers with a very high $\tau$, e.g. $\tau > \max_{j \in M, i \in N} d_j + v_j(i)$. In the aforementioned case, any transaction will, unambiguously, leave both agents engaging in it strictly worse off. Since no agents will engage in trade, the strategic situation becomes equivalent to plain RSD. Another immediate result is that even in the presence of transaction costs, allowing for trades from any fixed endowment leads to a Pareto-improvement, as any trades that happen must still be mutually beneficial.

Continuing with the interpretation of transaction costs as taxes, a canonical result in the analysis of taxation is that levying a tax on buyers or sellers does not change the outcome. Intuitively, that result is true as the market price would shift to account for the tax, and that shift depends only on the price elasticity of the good being taxed. In our current setting, we have the following:

\begin{theorem}
    Consider two agents $j$ and $j'$. Suppose agent $j$ can choose between keeping item $i$ or choosing item $\hat{i}$ and trading with agent $j'$ to receive their item of $i'$. Then, the trade happens if and only if 
    $$v_j(i') + v_{j'}(\hat{i}) - v_{j'}(i') - \tau > v_j(i)$$
    regardless of how the transaction cost is split between agents $j$ and $j'$.
\end{theorem}

\begin{proof}
    Suppose agent $j$ pays $c$ of the transaction cost and agent $j'$ then pays $\tau-c$ of the cost. Let $t$ be the transfer from agent $j$ to agent $j'$. For agent $j'$ to accept the trade, it must be that
    $$v_{j'}(\hat{i}) - (\tau - c) + t \geq v_{j'}(i') \implies t \geq v_{j'}(i') + (\tau-c) - v_{j'}(\hat{i}).$$
    Then, agent $j$ would prefer to offer a trade if and only if there is some $t$ that satisfies agent $j'$'s incentive compatibility constraint and also satisfies
    $$v_j(i') - c - t > v_j(i).$$
    Substituting in the minimal transfer to make agent $j'$ willing to trade gives that
    $$v_j(i') - c - [v_{j'}(i') + (\tau-c) - v_{j'}(\hat{i})] = v_j(i') +v_{j'}(\hat{i})- v_{j'}(i')- [\tau - c+c] = v_j(i') +v_{j'}(\hat{i})- v_{j'}(i')- \tau > v_j(i)$$
    as a necessary and sufficient condition for a trade to happen. This is exactly the condition outlined in the Theorem and does not depend on $c$.
\end{proof}

Using this, we get the following:

\begin{theorem}
    For sufficiently small $\tau > 0$, the allocation in the presence of transaction costs is the same as the allocation without transaction costs.
\end{theorem}

\begin{proof}
    Without transaction costs, trades happen if 
    $$v_j(i') + v_{j'}(\hat{i}) - v_{j'}(i') > v_j(i).$$
    As the inequality is strict, for each trade $\gamma \in \Gamma$ there exists some $\epsilon_\gamma > 0$ such that 
    $$v_j(i') + v_{j'}(\hat{i}) - v_{j'}(i') - \epsilon_\gamma > v_j(i).$$
    As there are a finite number of trades,
    $$\gamma^* = \min_{\gamma \in \Gamma} \epsilon_\gamma$$
    exists and is strictly positive, so taking $\tau \leq \gamma^*$ still makes every trade go through.
\end{proof}

We add to our theoretical analysis by incorporating the feature of transaction costs into a simulated experiment. We implement this feature in the simulation in two ways: in one instance, as a fixed sum that an agent must pay if they want to engage in a pairwise trade, and in a second instance as a sum proportional to the dollar amount of the transaction that must be paid in order to engage in the trade. The results of this experiment are shown in figure 3 and 4.

\begin{figure}[h!]
    \centering
    \includegraphics[width=0.6\textwidth]{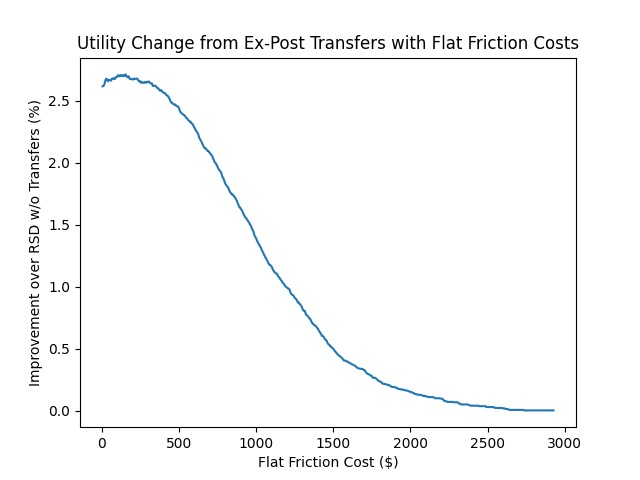}
    \caption{Utility Increases over Plain RSD for Increasing Constant Transaction Costs}
    \label{fig:friction-costs-constant}
\end{figure}

\begin{figure}[h!]
    \centering
    \includegraphics[width=0.6\textwidth]{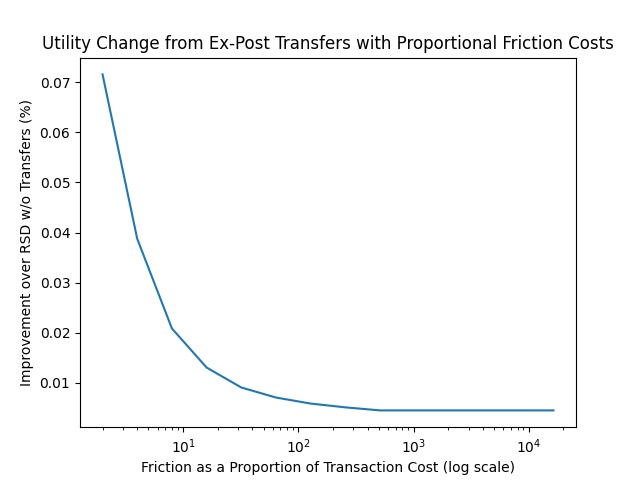}
    \caption{Utility Increases over Plain RSD for Increasing Proportional Transaction Costs}
    \label{fig:friction-costs-proportional}
\end{figure}

Our simulated results mirror closely our theoretical analysis: we see that once transactions costs reach a prohibitively high threshold, almost all trade ceases to be mutually beneficial, and the mechanism devolves to become equivalent to plain RSD. Furthermore, we also see that for sufficiently low transaction costs, mechanism functioning is essentially identical to a $\tau = 0$ scenario. Furthermore, the simulation results clarify how important low transaction costs are for the efficient functioning of RSD with ex-post transfers.

\section{Discussion and Real World Considerations}

Despite the interesting theoretical properties and abstract experimental results we can find about RSD with ex post transfers, ultimately, this system is meant to function in real world social environments. Thus, we use this section to describe important real-world considerations for the use of our mechanism to solve meaningful assignment problems. We return here to the motivating example and application of college housing assignments.

The adoption of RSD with ex post transfers may stir up controversy for a number of reasons.

At first glance, it can appear that administrators are 'making college even more expensive than it already is,' which could lead to push back from the press and the community. That being said, a deeper analysis, based on our theoretical and experimental results, reveals that this is mostly a 'PR problem:' our theoretical analysis shows how RSD with ex post transfers (with group trades) leads to higher welfare than current plain RSD mechanisms, and our simulation results show an increase in total welfare from use of RSD with ex post pairwise transfers in a more realistic setting.

Another, perhaps more salient concern is that of the advantage of high-resourced participants to outbid others when trying to purchase rooms they desire. Although these trades are Pareto-improvements (as shown earlier), this could promote a sentiment of the system being 'unfair' and biased against those without lots of money to spend in the transfer market, leading to backlash and dissatisfaction from the community that isn't captured by the utility functions we use in our model.

Furthermore, since the most popular and desired rooms tend to be located in the same area and high resource participants are generally willing and able to buy those rooms, this may lead to a kind of socio-economic segregation emerging in student housing patterns. This could lead to meaningful but hard to measure detriments to students' college experiences, like the loss of a chance to interact with and learn from those who aren't in one's own socio-economic group and, thus, people one would otherwise be less likely to interact with. Hence, the concern about high-resourced participants (or, generally, large inequality in initial resources of participants) could lead to meaningful ramifications that are hard to predict and hard to measure. That being said, these concerns are mitigated by our simulation results, which demonstrate that the benefit conferred upon higher resourced participants quickly hits sharply diminishing returns, leaving them similarly advantaged to most others in the market except those with exceptionally low initial endowments.

Additionally, the limited rationality of participants may lead to sub-optimal outcomes from using RSD with ex post transfers.

As both personal experience and behavioral economics bear out, people not only act irrationally but, to borrow a phrase, are 'predictably irrational' (\cite{ariely_2009_predictably}). Specifically, people tend to underweight the importance of the future and overweight the importance of the present, otherwise known as hyperbolic discounting or present bias ( \cite{chakraborty_2019_present}).

This could prove to be an important factor in how people behave in the ex-post transfer market, by under-weighting the importance of getting a good room, whose consequences span a year and are rather diffuse, and over-weighting how much (more) money they can leave the mechanism with, which is immediate and straightforward in its benefit. This is a unique concern of RSD with transfers: plain RSD has no 'temptation' of making money on transfers to cause myopia on the part of participants when choosing during RSD.

Reflecting on anecdotal evidence and personal experience, we observe that many undergrads will vehemently fight to secure rooms  they desire, communicating the seriousness and clarity that participants typically approach the room assignment process with. Thus, the concern of limited rationality seems itself limited in its plausibility and importance.

Finally, the issue of real world transaction costs in the use of the mechanism pose an important but tractable problem. From our simulation results, we see that once transaction costs reach a sufficiently low threshold, the mechanism is able to function efficiently and realize meaningful gains over plain RSD. Though transaction costs, like finding a person willing to sell their room and agreeing on secure way to transact, could be very high, these problems are also effectively addressed by support from a centralized mechanism. An officially-sanctioned online marketplace, where agents can offer their rooms for sale, request to buy others' rooms for a proposed price, and transact with a verified online mechanism could dramatically cut transaction costs and allow efficient functioning of the mechanism.

\section{Conclusion}

Given most analyses of assignment problems with RSD use agent utility functions that only vary with the item assigned, we show that RSD is no longer Pareto optimal when we consider a more realistic quasilinear utility model for agent welfare. To address this revised problem, we introduce the mechanism of RSD with ex-post transfers.

We find that RSD with ex-post transfers, theoretically with no constraints, leads to the socially optimal allocation of items, but that this property breaks down when limited to the realistic scenario of pairwise trading. That being said, even when limiting the transfer market to pairwise trading, we find in our simulations that our mechanism still outperforms plain RSD by a non-trivial margin in terms of total welfare generated. 

Importantly, we demonstrate RSD with ex-post transfers is not strategyproof, and analyze the properties of the strategic situation of the mechanism. In particular, agents must incorporate their beliefs about what will eventually happen in the transfers phase to pick the item that gives them the highest expected value. However, there can be significant uncertainty over others' actions and willingness to trade, making it extremely difficult for agents to successfully strategize when facing RSD with ex-post transfers. 

To overcome this issue, we introduce RSD with interim transfers. This mechanism has two clear improvements over RSD with ex-post transfers. First, restricting trades to being offered one at a time reduces strategic uncertainty by removing the issue of whether or not an agents will end up choosing a better trade. Second, endogenizing trades as part of the mechanism instead of leaving it up to a messy aftermarket may reduce search frictions, decreasing transaction costs. Unfortunately, as our analysis of the case in which there are two agents and items, optimal strategizing in RSD with interim transfers is generally complicated and requires further investigation. Some questions for future research include:
\begin{enumerate}
    \item What are some comparative statics results that characterize behavior?
    \item Are there any conditions that can be added to induce some notion of ``straightforward'' strategies?
    \item What happens if valuations are correlated or if all agents have a common prior that values are drawn from?
    \item If agent valuations can be pinned down by a one-dimensional type as a measure of how much they care about some item characteristic, can envelope-theorem type analyses pin down equilibrium expected utility?
\end{enumerate}

A key consideration in the implementation of the transfer period is the transaction costs associated with engaging in trade in the market. We find in our simulations that transaction costs significantly effect transfer market efficiency, but that sufficiently low transaction costs are tractable with the use of a centralized marketplace mechanism, leading to RSD with ex-post transfers functioning efficiently.

Finally, we find that only a few real world considerations are salient in the implementation of our mechanism. For example, though  resource inequality of participants can create unintended negative consequences, we find the direct influence of wealth inequality to be quite limited in terms of agents' power in and success from participating in the mechanism.

Overall, we find that RSD with ex-post transfers is a relatively simple augmentation to the popular RSD mechanism which neatly addresses the assignment problem of agents with quasilinear utilities, exhibiting interesting strategic characteristics, theoretical and simulated efficiency, and real-world promise.

\newpage 

\bibliography{cites.bib}

\newpage

\appendix
\begin{center}
    \textbf{\Large Appendices}
\end{center}

\section{Model Cards}

For the sake of clarity, we summarise the range of models we analyse here.

Our baseline model is Random Serial Dictatorship.

\begin{mech}{(Random) Serial Dictatorship}
    Given a set of items $N$, a set of agents $M$, and each agent (privately) has preferences represented by $u_j$, do the following:
    \begin{enumerate}
        \setlength\itemsep{0.5em}
        \item Order the agents $M$ based in some arbitrary order;
        \item For each agent $j$ in $M$:
        \begin{enumerate}
            \item Ask agent $j$ to choose $i_j \in I$ maximizing their utility: $i_j \in \arg\max u_j(i,d_j)$;
            \item Set $a(j) = i_j$;
            \item Remove $i_j$ from $I$.
        \end{enumerate}
        \item Repeat until there are no agents or items left and set $a(j) = \varnothing$ for any remaining agents;
        \item Set $t(j) = 0$ for all agents.
    \end{enumerate}
    \vspace{0.5em}
    If the order in step one is by random, we call the mechanism Random Serial Dictatorship (RSD).
\end{mech}

Our key contribution to the literature comes in the next model, which is RSD with ex-post transfers.

\begin{mech}{(Random) Serial Dictatorship with Ex-Post Transfers}
    Given a set of items $N$, a set of agents $M$, and each agent (privately) has preferences represented by $u_j$, do the following:
    \begin{enumerate}
    \setlength\itemsep{0.5em}
        \item Order the agents $M$ based in some arbitrary order and set $I_1 = I$;
        \item For each agent $j$ in $M$:
        \begin{enumerate}
            \item Ask agent $j$ to choose $i_j \in I_j$;
            \item Set $e(j) = i_j$;
            \item Set $I_{j+1} = I_j \setminus \{i_j\}$.
        \end{enumerate}
        \item Repeat until there are no agents or items left and set $a(j) = \varnothing$ for any remaining agents;
        \item Set $t(j) = 0$ for all agents;
        \item Allow agents to trade until general equilibrium is reached. If $p:N \to \mathbb{R}$ maps items to their price in general equilibrium and $a$ is the allocation of goods in general equilibrium, set $t(j) = p(e(j)) - p(a(j))$ to be the value of agent $j$'s endowment minus the cost of agent $j$'s house in equilibrium.
    \end{enumerate}
\end{mech}

Because in practice, we do not see ex post transfers occur in general equilibrium, our third model is one which looks only at pairwise transfers after RSD.

\begin{mech}{(Random) Serial Dictatorship with Ex-Post Pairwise Transfers}
    Given a set of items $N$, a set of agents $M$, and each agent (privately) has preferences represented by $u_j$, do the following:
    \vspace{0.5em}
    \begin{enumerate}
    \setlength\itemsep{0.5em}
        \item Order the agents $M$ based in some arbitrary order and set $I_1 = I$;
        \item For each agent $j$ in $M$:
        \begin{enumerate}
            \item Ask agent $j$ to choose $i_j \in I_j$;
            \item Set $e(j) = i_j$;
            \item Set $I_{j+1} = I_j \setminus \{i_j\}$.
        \end{enumerate}
        \item Repeat until there are no agents or items left and set $a(j) = \varnothing$ for any remaining agents;
        \item Set $t(j) = 0$ for all agents;
        \item Allow agents to trade. If $k$ is willing to pay $p$ to trade with $j$ then set:
         \begin{itemize}
            \item $a'(k) = a(j)$ and $a'(j) = a(k)$ where $a'$ is the new allocation after the trade;
            \item $t'(k) = t(i) - p$ and $t'(j) = t(k) + p$ where $t'$ is the new transfer after the trade.
        \end{itemize}
    \end{enumerate}
\end{mech}

Finally, we consider what happens when we do transfers during the RSD process, which is a useful model for when the RSD process happens over a sufficiently long span of time as to allow trades during its occurrence.

\begin{mech}{(Random) Serial Dictatorship with Interim Transfers}
    Given a set of items $N$, a set of agents $M$, and each agent (privately) has preferences represented by $u_j$, do the following:
    \vspace{0.5em}
    \begin{enumerate}
    \setlength\itemsep{0.5em}
       \item Order the agents $M$ based in some arbitrary order, set $t(j) = 0$ for all agents, and set $I_1 = I$;
        \item For each agent $j$ in $M$:
        \begin{enumerate}
            \item Ask agent $j$ to choose $i_j \in I_j$;
            \item Set $a(j) = i_j$;
            \item Set $I_{j+1} = I_j \setminus \{i_j\}$;
            \item Allow agent $j$ to offer a trade with any agent $j' < j$ at a price of $t$;
            \begin{itemize}
                \item If the trade is accepted, swap $a(j)$ and $a(j')$ and set $t(j) = t(j)-t, t(j') = t(j')+t$;
                \item If the trade is not accepted or agent $j$ does not offer a trade, do nothing.
            \end{itemize}
        \end{enumerate}
        \item Repeat until there are no agents or items left and set $a(j) = \varnothing$ for any remaining agents;
    \end{enumerate}
\end{mech}

\newpage

\section{Omitted Proofs}

\noindent \textbf{PROOF OF LEMMA \ref{trade1}:}

\begin{proof}
    Suppose $v_j(a(k)) + v_k(a(j)) > v_j(a(j)) + v_k(a(k))$. Then, 
    $$v_j(a(k)) - v_j(a(j)) > v_k(a(k)) - v_k(a(j))$$
    so there exists some $\epsilon > 0$ such that 
    $$v_j(a(k)) - v_j(a(j)) = v_k(a(k)) - v_k(a(j)) + \epsilon$$
    or equivalently, 
    $$\epsilon = \left[v_j(a(k)) + v_k(a(j))\right] - \left[v_j(a(j)) + v_k(a(k))\right].$$
    Next, set $t(j) = -\epsilon/2$ and $t(k) = \epsilon/2$ so
    \begin{equation*}
        \begin{split}
            v_j(a(k)) + t(j) - v_j(a(j)) &= v_j(a(k)) - \frac{\left[v_j(a(k)) + v_k(a(j))\right] - \left[v_j(a(j)) + v_k(a(k))\right]}{2} - v_j(a(j)) \\
            &= \frac{2v_j(a(k)) -\left( \left[v_j(a(k)) + v_k(a(j))\right] - \left[v_j(a(j)) + v_k(a(k))\right] \right)- 2v_j(a(j))}{2}\\
            &= \frac{[v_j(a(k))+v_k(a(j))] - [v_j(a(j))+v_k(a(k))]}{2}\\
            &> 0
        \end{split}
    \end{equation*}
    so $v_j(a(k)) + t(j) > v_j(a(j))$ and similar reasoning shows that $v_k(a(j)) + t(k) > v_k(a(k))$. Thus, both agents are better off and a trade is possible. 

    Conversely, suppose $v_j(a(k)) + v_k(a(j)) \leq v_j(a(j)) + v_k(a(k))$. Towards a contradiction, suppose there is some $p$ such that setting $t(j) = p, t(k) = -p$ makes $v_k(a(j)) + t(k) > v_k(a(k))$ and $v_j(a(k)) +t(j) > v_j(a(j))$. However, this implies
    $$v_k(a(j)) + t(k) + v_j(a(k)) +t(j) = v_k(a(j)) -p + v_j(a(k)) + p  = v_k(a(j)) +  v_j(a(k)) > v_j(a(j)) + v_k(a(k))$$
    which is a contradiction.
\end{proof}

\noindent \textbf{PROOF OF THEOREM \ref{interim}}
\begin{proof}
    We will prove this by induction on $j$ where $j$ ranges from $1$ to $m$. In particular, we will show for all $j$, we have that the allocation function restricted to the agents that have already chosen $a|_j: \{1,2,...,j\} \to I$ maximizes
    $$\sum_{j' = 1}^j v_{j'}(a|_j(j')).$$
    Going forward, we will suppress the restriction and just write $a$ instead of $a|_j$ with the understanding that $a$ has not yet been defined on the full domain.

    The base case. Suppose $j = 1$ and only the first agent has chose. Clearly, they cannot offer any trades since there is no one to trade with so $1$'s best option is to just choose their favorite item. As such, $\sum_{j' = 1}^1 v_{j'}(a(j')) = v_1(a(1))$ is maximized.

    The inductive step. Suppose that after agent $j$'s choice, $a$ maximizes
    $$\sum_{j' = 1}^j v_{j'}(a(j')).$$
    Towards a contradiction, suppose that after agent $j+1$'s choice, $a$ no longer maximizes
    $$\sum_{j' = 1}^{j+1} v_{j'}(a(j')).$$
    There are two possible cases: either agent $j+1$ made a trade or did not.
    \begin{enumerate}
        \item If agent $j+1$ did not make a trade, then we can write
        $$v_{j+1}(a(j+1)) = \sum_{j' = 1}^{j+1} v_{j'}(a(j')) - \sum_{j' = 1}^j v_{j'}(a(j')).$$
        As $\sum_{j' = 1}^j v_{j'}(a(j'))$ is independent of $a(j+1)$, we have that $\sum_{j' = 1}^{j+1} v_{j'}(a(j'))$ not being maximized implies that there exists some $a'$ such that $a(j') = a'(j')$ for $j' \leq j$ and $a(j+1) \neq a'(j+1)$. Then,
        \begin{equation*}
            \begin{split}
                \sum_{j' = 1}^{j+1} v_{j'}(a'(j')) > \sum_{j' = 1}^{j+1} v_{j'}(a(j')) &\implies \sum_{j' = 1}^{j+1} v_{j'}(a'(j')) - \sum_{j' = 1}^j v_{j'}(a(j')) > \sum_{j' = 1}^{j+1} v_{j'}(a(j')) - \sum_{j' = 1}^j v_{j'}(a(j')) \\
                &\implies v_{j+1}(a'(j+1)) > v_{j+1}(a(j+1))
            \end{split}
        \end{equation*}
        so agent $j+1$ choosing $a(j+1)$ over $a'(j+1)$ could not have been played in equilibrium.
        \item Suppose under allocation $a$, agent $j+1$ chooses item $i$ and trades with agent $j_1$ who has item $i_1$. Next, suppose under allocation $a'$, the difference is agent $j+1$ chooses item $i'$ and trades with agent $j_2$ who has item $i_2$. We will show that if allocation $a$ maximizes total welfare for agents up to $j+1$, then agent $j+1$ prefers allocation $a$ (and its associated transfers) to allocation $a'$ (and its associated transfers). By the feasibility assumption, allocation $a$ is achievable by agent $i$.

        Since $a$ maximizes total utility and $a' \neq a$, we have that
        $$\sum_{j' = 1}^j v_{j'}(a(j')) > \sum_{j' = 1}^j v_{j'}(a'(j')).$$
        As all agents other than agents $j_1,j_2,j+1$ are allocated the same item under $a$ and $a'$, it must be that
        $$\sum_{j' = j_1,j_2,j+1} v_{j'}(a(j')) > \sum_{j' = j_1,j_2,j+1} v_{j'}(a'(j'))$$
        since total utility for agents other than these three is constant. By construction, this means that 
        $$v_{j_1}(i) + v_{j_2}(i_2) + v_{j+1}(i_1) > v_{j_1}(i_1) + v_{j_2}(i') + v_{j+1}(i_2)$$
        which then yields
        $$v_{j+1}(i_1) - v_{j_1}(i_1) + v_{j_1}(i) > v_{j+1}(i_2) - v_{j_2}(i_2) + v_{j_2}(i')$$
        after rearranging terms. Note that agent $j+1$ needs to pay agent $j_1$ at least $t_1 \geq v_{j_1}(i_1) - v_{j_1}(i)$ to make $j_1$ willing to trade (getting to allocation $a$) while agent $j+1$ needs to pay agent $j_2$ at least $t_2 \geq v_{j_2}(i_2) - v_{j_2}(i')$ to make $j_2$ willing to trade (getting to allocation $a'$). As such, agent $j+1$'s utility is upper bounded by $v_{j+1}(i_1) - v_{j_1}(i_1) + v_{j_1}(i)$ if allocation $a$ is achieved, which is greater than the upper bound of $v_{j+1}(i_2) - v_{j_2}(i_2) + v_{j_2}(i')$ if allocation $a'$ is reached. As such, agent $j+1$ would prefer to choose the item that reaches the allocation that maximizes total welfare over any other item; once again, for any fixed way to split gains from trade, the preference is strict.
    \end{enumerate}
    Thus, the inductive step holds in either case and we get the desired result.
\end{proof}

\newpage

\noindent \textbf{PROOF OF THEOREM \ref{compstat}}
\begin{proof}
    First, rewrite $\pi_2(t)$ as
    \begin{equation*}
        \begin{split}
            \pi_2(t) &= (v_2(A)-t) \cdot \left[1 - \int_{\underline{V}}^{\overline{V}} F_1^B(s - t)dF_1^A(s)\right] + v_2(B) \cdot \int_{\underline{V}}^{\overline{V}} F_1^B(s - t)dF_1^A(s)\\
            &= (v_2(A)-t) - (v_2(A)-t) \cdot \int_{\underline{V}}^{\overline{V}} F_1^B(s - t)dF_1^A(s) + v_2(B) \cdot \int_{\underline{V}}^{\overline{V}} F_1^B(s - t)dF_1^A(s)\\
            &= (v_2(A)-t) + \left[ v_2(B) - (v_2(A)-t)\right]\cdot \int_{\underline{V}}^{\overline{V}} F_1^B(s - t)dF_1^A(s). 
        \end{split}
    \end{equation*}

Notice that the leading $v_2(A)$ does not depend on $t$, so that leading term can be removed without affecting $2$'s optimization problem. Next, let $d = v_2(B) - v_2(A)$ be the difference in $2$'s valuations of item $A$ and $B$. Define the parameterized objective function $\pi_2'(t;d)$ by
$$\pi_2'(t;d) =  (d+t) \cdot \int_{\underline{V}}^{\overline{V}} F_1^B(s - t)dF_1^A(s) - t.$$
This is what player $2$ chooses $t$ to maximize taking in $d$ as a parameter. 

Next, suppose $t' > t$ and $d' > d$. We then have
\begin{equation*}
    \begin{split}
        \pi_2'(t';d') - \pi_2'(t;d') =& \left[(d'+t') \cdot \int_{\underline{V}}^{\overline{V}} F_1^B(s - t')dF_1^A(s) - t'\right]-\left[(d'+t) \cdot \int_{\underline{V}}^{\overline{V}} F_1^B(s - t)dF_1^A(s) - t\right]\\
        =& d'\left[ \int_{\underline{V}}^{\overline{V}} F_1^B(s - t')dF_1^A(s) - \int_{\underline{V}}^{\overline{V}} F_1^B(s - t)dF_1^A(s)\right] \\
        & + t' \int_{\underline{V}}^{\overline{V}} F_1^B(s - t')dF_1^A(s) - t\int_{\underline{V}}^{\overline{V}} F_1^B(s - t)dF_1^A(s) - t'+t' \\
        <& d\left[ \int_{\underline{V}}^{\overline{V}} F_1^B(s - t')dF_1^A(s) - \int_{\underline{V}}^{\overline{V}} F_1^B(s - t)dF_1^A(s)\right] \\
        & + t' \int_{\underline{V}}^{\overline{V}} F_1^B(s - t')dF_1^A(s) - t\int_{\underline{V}}^{\overline{V}} F_1^B(s - t)dF_1^A(s) - t'+t'\\
        =& \left[(d+t') \cdot \int_{\underline{V}}^{\overline{V}} F_1^B(s - t')dF_1^A(s) - t'\right]-\left[(d'+t) \cdot \int_{\underline{V}}^{\overline{V}} F_1^B(s - t)dF_1^A(s) - t\right]\\
        &= \pi_2'(t';d) - \pi_2'(t;d) 
    \end{split}
\end{equation*}
where the inequality holds since $d' > d$ by assumption and $t' > t$ implies $F_1^B(s - t') < F_1^B(s - t)$ for all $s$ by properties of a cumulative density function. As such, $\pi_2'$ exhibits strict decreasing differences so by the Strict Monotonicity Theorem (Paul Milgrom's lecture slides; see \cite{milgrom_1994_monotone} for a more general version of this result) we get that the optimal choice of $t$ to maximize $\pi'_2(t;d)$ is weakly decreasing as $d$ increases. However, $d = v_2(B) - v_2(A) = -(v_2(A) - v_2(B))$ so as $v_2(A) - v_2(B)$ grows, $d$ decreases and hence the optimal choice of $t$ is weakly increasing.
\end{proof}

\end{document}